\documentclass{amsart}

\usepackage{amsthm,amssymb,verbatim,amsmath, amsfonts, amscd}
\usepackage{xcolor}

\usepackage{amssymb}
\usepackage{amsthm}

\usepackage{amsthm,amssymb,verbatim,amsmath, amsfonts, amscd 
,a4wide
,color}
\usepackage{ascmac}    

\usepackage{verbatim} 

\theoremstyle{plain}\newtheorem{Theorem}{Theorem}[section]
\theoremstyle{plain}\newtheorem{Conjecture}[Theorem]{Conjecture}
\theoremstyle{plain}\newtheorem{Corollary}[Theorem]{Corollary}
\theoremstyle{plain}\newtheorem{Lemma}[Theorem]{Lemma}
\theoremstyle{plain}\newtheorem{Proposition}[Theorem]{Proposition}
\theoremstyle{definition}
\theoremstyle{definition}
\theoremstyle{definition}
\theoremstyle{definition}\newtheorem{Remark}[Theorem]{Remark}
\theoremstyle{definition}
\theoremstyle{definition}
\theoremstyle{definition}
\theoremstyle{definition}
\theoremstyle{definition}
\theoremstyle{definition}
\theoremstyle{definition}
\theoremstyle{definition}
\theoremstyle{definition}\newtheorem{Notation/Definition}
[Theorem]{Notation/Definition}
\theoremstyle{definition}

\def\CB{{\mathcal{B}}}

\def\PAut{\mathrm{PAut}}

\def\dim{\mathrm{dim}}          
\def\stab{\mathrm{stab}}

\def\wt{\mathrm{wt}} 
\def\min{\mathrm{min}}

\newcommand{\FF}{{\mathbb{F}}}

\begin{document}

\title{Involutory permutation automorphisms of binary linear codes}
\date{\today}
\author{{Fatma Altunbulak Aksu, Roghayeh Hafezieh and {\.I}pek Tuvay}}
\address{Mimar Sinan Fine Arts University, Department of Mathematics, 34380, Bomonti, \c{S}i\c{s}li, Istanbul, Turkey}
\email{fatma.altunbulak@msgsu.edu.tr}
\address{Gebze Technical University, Department of Mathematics, 41400, Gebze, Kocaeli, Turkey}
\email{roghayeh@gtu.edu.tr}
\address{Mimar Sinan Fine Arts University, Department of Mathematics, 34380, Bomonti, \c{S}i\c{s}li, Istanbul, Turkey}
\email{ipek.tuvay@msgsu.edu.tr}

\thanks{}


\begin{abstract}
{ We investigate the properties of binary linear codes of even length whose permutation 
automorphism group is a cyclic group generated by an involution. Up to dimension or co-dimension $4$, 
we show that there is no quasi group code 
whose permutation automorphism group is isomorphic to $C_2$. By generalizing the method 
we use to prove this result, we obtain results on the structure of putative extremal self-dual $[72, 36, 16]$ 
and $[96, 48, 20]$ codes in the presence of an involutory permutation automorphism.}
\end{abstract}

\maketitle

\section{Introduction}

A {\it binary linear code $C$ of length $n$} is a subspace of $\FF_2^n$. The symmetric group $S_n$ 
acts on $\FF_2^n$ by permuting the set of its coordinates. The group of permutations which sends $C$ onto itself  
is called {\it the permutation automorphism group of $C$} and is denoted by $\PAut(C)$. The 
permutation automorphism group of a code carries a lot of information about the algebraic structure of 
the code. For instance, 
if there is a non-trivial $G \leq \PAut(C)$ for a binary linear code $C$, 
the problem of determining $C$ by the help of $G$ is a well-known and difficult problem in coding theory 
(see for example \cite{Bor2015}).

Our intention to attack this problem is to investigate the structure of binary linear codes in the 
presence of an involutory permutation automorphism. In \cite{Bou2000}, a method for constructing self-dual 
codes with an involutory permutation automorphism group is given. Later in \cite{Bou2002}, 
extremal doubly even self-dual binary codes of length $24m$ with an involutory permutation automorphism 
are considered. It is shown in \cite[Theorem 5.3]{Bou2002} that a putative self-dual $[72, 36, 16]$ code 
can not have an involutory automorphism with some fixed points. 
Moreover, there are five candidates for the permutation automorphism group of 
a putative self-dual $[72, 36, 16]$ code and two of the cantidates are $C_2$ and $C_2 \times C_2$ (see \cite[Theorem 6.3]{Bor2015}). 
Inspired by all of these, 
our aim is to deduce some structural results on binary linear codes having an involutory permutation 
automorphism. 
We first discuss linear codes with an arbitrary involutory permutation automorphism, then we focus 
on the ones with involutory fixed point free permutation automorphisms. The fixed subcode associated to 
a permutation automorphism gives a lot of information about the code itself, as we can see in \cite{Bou2000, Bou2002, BC2018}
In particular, the fixed subcode of an involutory permutation automorphism plays a central role in our arguments.

Borello and Willems, in \cite{BW2022}, give an intrinsic 
description of linear codes of length $n$ whose permutation automorphism group contains a free subgroup of $S_n$. These 
codes are called {\it quasi group codes}. Abelian codes, group codes, quasi-cyclic 
and quasi-abelian codes belong to the class of quasi group codes. Hence, this approach gives a 
unified way to consider these different class of codes. In this paper, we show that there is no quasi group 
code $C$ of even length $n>2$ and dimension $k$ or $n-k$ where $1\leq k \leq 4$ where $\PAut(C)\cong C_2$ (see Corollary 
\ref{quasi}). On the other hand, by the methods we develop to prove this result, we obtain some results on 
the structures of putative extremal self-dual $[72, 36, 16]$ and $[96, 48, 20]$ codes.

The organization of the paper is as follows. In Section $2$, we give the background and some lemmas 
concerning the dimension of the fixed subspace of an involutory permutation automorphism. 
In Section $3$, we show that if a binary linear code $C$ of even length 
with dimension or co-dimension $2$ satisfies $\PAut(C)\cong C_2$, then its length is $2$ or $4$. We also 
determine all binary linear codes $C$ of length $4$ with $\PAut(C)\cong C_2$ up to permutation equivalence.
Section $4$ deals with binary linear codes of even length with involutory fixed point free permutation automorphisms 
with dimension $3$ and dimension $4$. Also, there are some important corollaries together with a conjecture 
in this section. In Section $5$, 
we give a generalization of our idea that is used in the proof of Theorem \ref{dim4} and we present the applications on 
putative extremal self-dual $[72, 36, 16]$ and $[96, 48, 20]$ codes.

\section{Background and lemmas}

In this paper, we deal with binary linear codes. Recall that a {\it binary linear code $C$ of length 
$n$} is a subspace of $\FF_2^n$. An element $c=c_1 c_2 \ldots c_n \in C$ is called a {\it codeword of $C$} 
and its {\it Hamming weight} is defined to be 
$$\wt(c)=|\{ i\in \{1, \ldots, n\} \ | \ c_i\neq 0\}|.$$ 
The {\it minimum weight of $C$} is defined by $\wt(C)=\min \{ \wt(c) \ | \ c\in C\backslash \{0\} \}$. If $C$ is 
a binary linear code with length $n$, dimension $k$ and minimum weight $d$, then $C$ is called 
a {\it binary $[n, k, d]$ code}. For a code $C\subseteq \FF_2^n$ and a non-negative integer $i$, let $A_i(C)$ denote 
the number of codewords in $C$ of weight $i$. The sequence $(A_0(C), A_1(C), \ldots, A_n(C))$ is 
called {\it the weight distribution of $C$. }

There is a natural action of $S_n$ on $\FF_2^n$ induced by the action of $S_n$ on the set $\{1, \ldots, n\}$. More precisely, 
for $c=c_1 c_2 \ldots c_n \in \FF_2^n$ and $\beta \in S_n$, 
$$c^{\beta}=c_{\beta^{-1}(1)} c_{\beta^{-1}(2)} \ldots c_{\beta^{-1}(n)}.$$ 
For a linear code $C$, set $C^{\beta}=\{c^{\beta} \ | \ c\in C\}$. It is easy to see that 
$C^{\beta}$ is also a linear code. A permutation $\beta \in S_n$ is called a {\it permutation automorphism} 
or simply an {\it automorphism} of $C$ if $C^{\beta}=C$. The stabilizer  
$$\PAut(C)=\{\beta \in S_n \ | \ C^{\beta}=C \}$$ 
is called the {\it permutation automorphism group of $C$}.

For a linear code $C$ of length $n$, if there exists a subgroup $G$ of $\PAut(C)$ which is transitive 
and free, then $C$ is called a {\it group code} or a {\it $G$-code}. This characterization is given 
first in \cite{BRS2009} then in \cite{BW2022}. where a weaker notion is introduced. $C$ is called a {\it quasi group code} or a {\it quasi-$G$ code} if 
there is a free subgroup $G$ inside $\PAut(C)$. 

Two linear codes $C_1$ and $C_2$ of length $n$ 
are called {\it permutation equivalent} if there is a $\beta \in S_n$ such that $C_1^{\beta}=C_2$. In this case,
note that $\PAut(C_2)=\PAut(C_1)^{\beta}$. So if two linear codes are permutation equivalent, then their 
permutation automorphism groups are conjugate in $S_n$. In this paper, we denote the $k$-cycles in $S_n$ 
as $(a_1, a_2, \ldots, a_k)$ where $a_1, a_2, \ldots, a_k \in \{1, 2, \ldots, n\}$.

For a permutation automorphism $\beta$ of a binary linear code $C$,
the 
fixed point set of $\beta$, which we denote by $F_{\beta}(C)$, is defined as 
$$F_{\beta}(C)=\{ c \in C \ | \ c^{\beta}=c\}.$$
Note that, this set is closed under addition, so $F_{\beta}(C)$ is a subspace of $C$. This 
is called the {\it fixed subcode of $C$ under $\beta$}.

We begin with an observation that holds for any permutation automorphism of a binary linear code which is 
an involution. It turns out that there is a relationship between the dimension of the code and its fixed subcode 
under that involulion. 
The notation $\lceil x \rceil$ denotes the smallest integer greater or equal to $x$.

\begin{Lemma}\label{halfdim}
Let $C$ be a binary linear code of length $n$ and dimension $k$. Let $\beta \in \PAut(C)$ be an involution and 
$f=\dim(F_{\beta}(C))$. We have that $f\geq (k-f)$ or equivalently $f \geq \lceil k/2 \rceil$.
\end{Lemma}

\begin{proof}
If $C=F_{\beta}(C)$, the claim holds trivially. So assume that $C\neq F_{\beta}(C)$. 
Let's start with a basis $\CB_1$ of $F_{\beta}(C)$ and then extend it to a basis $\CB$ of $C$. 
Let $$\CB \backslash \CB_1= \{ w_{(1)}, \ldots, w_{(k-f)}\}$$
and for any $i=1, \ldots, k-f$ define $x_{(i)}=w_{(i)}+w_{(i)}{}^{\beta}$. Then since $w_{(i)} \not\in F_{\beta}(C)$, we have that 
$x_{(i)}\neq 0$. Also $x_{(i)}=x_{(i)}{}^{\beta}$, so that 
$x_{(i)} \in F_{\beta}(C)$. If for some $\lambda_i \in \FF_2$, $\sum_{i=1}^{k-f} \lambda_i x_{(i)}=0$, it follows that 
$$\sum_{i=1}^{k-f} \lambda_i w_{(i)}=(\sum_{i=1}^{k-f} \lambda_i w_{(i)})^{\beta}$$
and so $\sum_{i=1}^{k-f} \lambda_i w_{(i)} \in F_{\beta}(C) \cap \langle \CB \backslash \CB_1 \rangle$. 
However, $F_{\beta}(C) \cap \langle \CB \backslash \CB_1 \rangle =\{0\}$, which implies that 
$\sum_{i=1}^{k-f} \lambda_i w_{(i)}=0$. Hence by the linear independence of the set $\CB \backslash \CB_1$, it 
follows that $\lambda_i=0$ for $i=1, \ldots, k-f$. Therefore, $\{x_{(1)}, \ldots, x_{(k-f)} \}$ is a linearly 
independent subset of $F_{\beta}(C)$, which yields that $f\geq (k-f)$.
\end{proof}

\begin{Lemma}\label{dim3}
Suppose that $C$ is a binary linear code of dimension $k$ and length $n\geq 4$. If $\PAut(C)=\langle \beta \rangle$, 
for $\beta =(1,2)(3,4) \ldots (t, t+1)$ where $t$ is an odd integer such that 
$3 \leq t \leq (n-1)$ is an arbitrary odd integer, then $\dim(F_{\beta}(C))  \leq (k-1)$. In particular, 
if $k\geq 2$ we have that $\lceil k/2\rceil \leq \dim(F_{\beta}(C))  \leq (k-1)$.
\end{Lemma}

\begin{proof}
Let $f=\dim(F_{\beta}(C))$. If $f=k$, it follows that $F_{\beta}(C)=C$. Then every element $c$ of $C$ satisfy 
$c_i=c_{i+1}$ for all odd $i \in \{1, 2, \ldots, t\}$. But then the transposition $(i, i+1)\in \PAut(C)$ for any 
odd $i \in \{1, 2, \ldots, t\}$. This implies that $\PAut(C) \neq \langle \beta \rangle$, a contradiction, so 
$f \neq k$ and the result follows from Lemma \ref{halfdim}.

\end{proof}

\section{Codes of even length with lower dimensions or co-dimensions}

We begin with observations on permutation automorphisms of binary linear codes of even length 
whose dimension or co-dimension is equal to $1$. 

\begin{Proposition} \label{dim1}Let $C$ be a binary linear code of length $n=2 m$ where $m\geq 2$. Assume 
that $\dim(C)=1$ or $\dim(C)=n-1$. Then the following hold.
\begin{enumerate}
\item[\rm (i)] We have that $\PAut(C)\ncong C_2$. 
\item[\rm (ii)]  If $\wt(C)$ is even and $\wt(C)\neq n$, then $C$ is a quasi group code which is not a group code.
\item[\rm (iii)] Assume that $n=2^r$ (or equivalently $m=2^{r-1}$) and $\wt(C)\neq n$. Then $C$ is a quasi group code 
which is not a group code if and only if $\wt(C)$ is even.
\end{enumerate}

\end{Proposition}

\begin{proof} 
If $\dim(C)=n-1$, then $\dim(C^{\perp})=1$ and $\PAut(C)=\PAut(C^{\perp})$. Hence it is enough to 
prove all assertions for $1$-dimensional linear codes. 
Assume that $\dim(C)=1$ and the weight of $C$ is equal to $d$. It follows that 
$C$ is generated by $v$ where $v$ has exactly $d$ number of $1$'s. Then since $\PAut(C)=\stab_{S_n}(v)$, 
we have that $$\PAut(C)\cong S_d \times S_{n-d}.$$ Hence we have that $|\PAut(C)| > 2$.
So the assertion (i) is proved.

If $\wt(C)=d$ is even, there is a fixed point free involution $\gamma$ in $S_d$, and 
since $n-d$ is also even, there is a fixed point free involution $\delta$ in $S_{n-d}$. 
Then we have that $\langle \gamma \delta \rangle \leq \PAut(C)$ 
and it follows that $C$ is a quasi-$\langle \gamma \delta \rangle$ code. From the structure of $\PAut(C)$, it is easy to see 
that we can not find a transitive subgroup, hence $C$ is not a group code which finishes the proof of (ii).

Now let us prove the third assertion. Suppose that $C$ is a quasi-$H$ code for $H\leq \PAut(C)$. 
Then it follows that the order of $H$ is a divisor of $n=2^r$ and $H$ is a free 
subgroup of $S_n$. So there is a fixed point free involution $\gamma \in H$. It 
follows that $\gamma$ fixes $v$ and so the number of $1$'s 
in $v$ should be even since $\gamma$ is an involution. The converse direction 
follows from (ii).
\end{proof}

Now let us investigate linear codes with dimension or co-dimension equal to $2$.

\begin{Theorem}  \label{dim2} Let $C$ be a binary linear code of length $n=2 m$ where $m\geq 3$. 
If $\dim(C)=2$ or $\dim(C)=n-2$, then $\PAut(C)\ncong C_2$. 
\end{Theorem}

\begin{proof}
Suppose to the contrary that $C$ is a binary linear code of length $n \geq 6$ 
and $\PAut(C)= \langle \beta \rangle$ where $\beta$ is an involution in $S_n$. 
It is enough to prove the theorem when $\dim(C)=2$ since $\PAut(C)=\PAut(C^{\perp})$ 
and $\dim(C^{\perp})=n-\dim(C)$.

Since any conjugate of a permutation automorphism of $C$ is a permutation automorphism 
of the permutation equivalent code of $C$, it is enough to prove the result by considering all possible
involutions $\beta \in S_n$ with different cycle structures. Thus we can assume without loss of generality that $\beta =(1,2)(3,4) \ldots (t, t+1)$ where 
$1\leq t \leq (n-1)$ is an arbitrary odd integer. Let us first assume that
$3 \leq t \leq (n-1)$. Then by Lemma \ref{dim3}, the dimension $f$ of $F_{\beta}(C)$ is 
equal to $1$.
So $$C=\{0, x, w, w^{\beta} \}$$
where $x\in F_{\beta}(C)$ and $w\not \in F_{\beta}(C)$ and $w^{\beta}=w+x$. 
Since $x \in F_{\beta}(C)$, we know that $x_i=x_{i+1}$ for all odd $i$ where $i\leq t$. If there is an 
odd integer $j\leq t$ such that $x_j=x_{j+1}=0$, then since $w^{\beta}=w+x$, it implies 
that $w_j=w_{j+1}$. Hence the transposition $(j,j+1)$ fixes all elements on $C$, so it lies in $\PAut(C)$, 
which is a contradiction. If for any odd integer $i\leq t$ we have $x_i=x_{i+1}=1$, then since $w^{\beta}=w+x$, it implies 
that $w_i \neq w_{i+1}$ for any odd integer $i\leq t$. Then since $t \geq 3$, there exist two different integers $k, l \leq t+1$ 
such that $w_k=w_l$ so that $(w^{\beta})_k=(w^{\beta})_l$. Then it follows that the transposition $(k, l)$ fixes all 
elements in $C$, so $(k, l)\in \PAut(C)$. But since $(k, l)$ is different from $\beta$ this is a contradiction.  

Now suppose $t=1$, then $\beta=(1, 2)$. Suppose that $m\geq 4$. 
Let $u$ and $w$ be two non-zero elements of $C$. We claim 
that there exist at least two integers $k, l \in \{3, 4, \ldots, n\}$ such that $u_k=u_l$ and $w_k=w_l$. If this 
is not true, since $n\geq 8$, there exist at least three elements $k, t, l \in \{3, 4, \ldots, n\}$ such that 
$$u_k=u_t \text{ and } w_k\neq w_t, \ u_t=u_l \text{ and } w_t \neq w_l$$ 
But since these elements take only the values of $0$ and $1$, it would imply that $u_k=u_l \text{ and } w_k=w_l$. 
It follows that the transposition $(k, l)$ which is different from $\beta=(1, 2)$ fixes $u, w$ and $u+w$ hence is 
in $\PAut(C)$, which is a contradiction.

Suppose that $t=1$ and $m=3$, that is $\beta=(1, 2)$. 
Let $u$ and $w$ be two non-zero elements of $C$. 
If we can find $k, l \in \{3, 4, 5, 6\}$ such that $u_k=u_l$ and $w_k=w_l$ as in the previous paragraph then we are done. 
So let us assume that such $k$ and $l$ do not exist. Then it follows that there exist exactly two  integers $k, l \in \{3, 4, 5, 6\}$ 
with the property that $u_k=u_l=1$ and that $k', l' \in \{3, 4, 5, 6\}$ 
with the property that $u_{k'}=u_{l'}=0$. This implies that $w_k\neq w_l$ and $w_{k'}\neq w_{l'}$. 
If $u_1=1$, then $(1, k)$ and $(1, l)$ fixes $u$. 
Since $w_k\neq w_l$, we have that either $w_1=w_k$ or $w_1=w_l$. If $w_1=w_k$, then the transposition $(1, k)$ fixes 
$w$ and if $w_1=w_l$ then $(1, l)$ fixes $w$. 
If $u_1=0$. then both $(1, k')$ and $(1, l')$ fixes $u$.
Since $w_{k'}\neq w_{l'}$, we have that either $w_1=w_{k'}$ or $w_1=w_{l'}$, similarly 
choose either of the transposition $(1, k')$ or $(1, l')$ which fixes $w$ in this case. Since in either case we find 
a transposition in $\PAut(C)$ which is different from $\beta$, we arrive to a contradiction. This finishes the proof.

\end{proof}

\begin{Remark}
The lower bound for $m$ in Theorem \ref{dim2} is sharp. Indeed, we show that there exists a linear code of length $4$ whose 
permutation automorphism group is isomorphic to $C_2$ in the proof of the following result. When the length is 
equal to $2$, then $C=\FF_2^2$ and $\PAut(C)=C_2$ in this case.
\end{Remark}

\begin{Proposition} \label{characterization} Let $C$ be a binary linear code of length 
$4$ and $\dim(C)=2$. Let $\beta \in S_4$ be a transposition. Then $\PAut(C)=<\beta>$ if 
and only if for every $c\in C$ we have $c^{\beta}=c$ and $A_i(C)=1$ for $i=0,1,2,3$ and $A_4(C)=0$. 

\end{Proposition}

\begin{proof}
Without loss of generality we can assume that $\beta=(1, 2)$.

Suppose that $\PAut(C)=<\beta>$. We first claim that $F_{\beta}(C)=C$. If $F_{\beta}(C)\neq C$, then 
the dimension of $F_{\beta}(C)$ is equal to $1$. Let $x$ be the non-zero element of $F_{\beta}(C)$, then 
$$C=\{0, x, w, w^{\beta} \}$$ 
where $w\not\in F_{\beta}(C)$ and $x=w+w^{\beta}$. Since $w\not\in F_{\beta}(C)$, we have that $w_1 \neq w_2$, 
so that $x_1=x_2=1$. Moreover, since $x=w+w^{\beta}$ we have that $x_3=x_4=0$. Hence, $x=1100$. If $w_3=w_4$ then 
the transposition $(3, 4)$ fixes $w$ and $w^{\beta}$, as well as $x$. So we get a contradiction. So $w_3 \neq w_4$. Since 
we also have that $w_1 \neq w_2$, we have that either $w_1=w_3$ or $w_1=w_4$. Assume without loss of generality that 
$w_1=1$. Then if $w_1=w_3$ or $w_1=w_4$, then $(1, 3)$ or $(1, 4)$ fixes $w$ respectively. In the first case $(w^{\beta})_2=(w^{\beta})_3=1$, 
so that $(w^{\beta})^{(1, 3)}=1100=x$ and $x^{(1, 3)}=w^{\beta}$. In the second case $(w^{\beta})_2=(w^{\beta})_4=1$, 
so that $(w^{\beta})^{(1, 4)}=1100=x$ and $x^{(1, 4)}=w^{\beta}$. In both of the cases, we obtain a transposition in $\PAut(C)$ 
which is different from $\beta$. So we have that $F_{\beta}(C)=C$, or equivalently for all $c\in C, c^{\beta}=c$.

Our next claim is $1111\not \in C$. Assume to the contrary that $1111\in C$. Then there are non-zero elements 
$u, w \in C$ such that $u+w=1111$. There are $k, l \in \{2, 3, 4\}$ such that $u_k=u_l$. But since 
$u+w=1111$, we have that $w_k=w_l$. Thus $(k, l)$ which is different from $\beta$ lies in $\PAut(C)$, a contradiction. So 
$A_4(C)=0$. 

Since $c^{\beta}=c$ for all $c\in C$, we have that $c_1=c_2$ for all $c\in C$. If $c_1=c_2=0$ for all 
$c\in C$, then there is an element $v\in C$ such that $v_3 \neq v_4$ (because otherwise there are less than 4 elements). But this 
implies that there is $u\in C$ such that $u\neq v$ with $u_3\neq u_4$. Then the involution $(3, 4)$ flips $u$ and $v$ and it fixes $u+v$. 
So $(3, 4)\in \PAut(C)$, a contradiction. So there is at least one element $w\in C$ such that $w_1=w_2=1$. Then it follows that $w=1110$ or $w=1101$ since 
$1111 \not \in C$. If $w=1110$, then the other elements of $C$ are $1100$ and $0010$. If $w=1101$, then the other elements of $C$ are $1100$ and $0001$. 
Hence, $A_i(C)=1$ for $i=0,1,2,3$. 

Suppose that for any $c\in C$ we have $c^{\beta}=c$ and $A_i(C)=1$ for $i=0,1,2,3$ and $A_4(C)=0$. 
Then we have that $c_1=c_2$ for any $c\in C$ and $1111 \not\in C$. Then since there exists a unique element with 
weight $1$, there is an element $u\in C$ such that $u=0010$ or $u=0001$. Moreover, since there is a unique element with 
weight equal to $3$, we have that either $1110$ or $1101$ is in $C$. Then we have that $C=\{0000, 0010, 1100, 1110\}$ or 
$C=\{0000, 0001, 1100, 1101\}$. Now it is not difficult to see that $\PAut(C)=\langle \beta \rangle$.
\end{proof}

\section{Fixed point free automorphisms}

Throughout this section let $C$ be a binary linear code of length $n=2m$ and let 
$\sigma \in \PAut(C)$ be a fixed point free involutory permutation. Then without 
loss of generality we can assume 

\[\sigma= \prod_{1\leq i\leq n, i \text{ odd}}(i, i+1)=(1,2)(3,4) \ldots (n-1, n).\]
Then it is easy to see that $c=c_1 c_2\ldots c_n \in F_{\sigma}(C)$ if and only if 
$c_i=c_{i+1}$ for all odd $i \in \{1, 2, \ldots, n\}$. For an element $c \in F_{\sigma}(C)$ 
let us define the set of all odd integers in the support of $c$ as $T_c$, that is
 $$T_c:=\{1\leq i \leq n\ | \ i \text{ is odd and } c_i c_{i+1}=11 \}.$$

\begin{Lemma}\label{wx}
Let $C$ be a binary linear code of length $n=2m$ with $\sigma =(1,2)(3,4) \ldots (n-1, n) \in \PAut(C)$. Let $w\in C \backslash F_{\sigma}(C)$ and 
$y:=w+w^{\sigma} \in F_{\sigma}(C)$. Assume that $x$ is a non-zero element in $F_{\sigma}(C)$ such that $T_x \subseteq T_y$, 
then the involution $\alpha_x= \Pi_{i \in T_x} (i, i+1)$ in $S_n$ satisfies $w^{\alpha_x}=w+x.$
\end{Lemma}

\begin{proof}
Since $w^{\sigma}=w+y$ and $T_x \subseteq T_y$, it is easy to see that if $i \in T_x$ then $x_ix_{i+1}=11$ and $w_i\neq w_{i+1}$.
So $(w+x)_i (w+x)_{i+1}=w_{i+1}w_i$ when $i\in T_x$. 
Moreover, if $i\not\in T_x$ then $x_i x_{i+1}=00$, so that $(w+x)_i (w+x)_{i+1}=w_{i}w_{i+1}$. As a result, for 
any odd $i$, we have that
$$(w+x)_i (w+x)_{i+1}= \begin{cases}
  w_{i+1} w_i  & \text{ if } i \in T_x \\
  w_i w_{i+1} & \text{ if } i\not\in T_x
\end{cases} .$$
On the other hand, by the definition of $\alpha_x$, 
for any odd $i$ we have that
$$(w^{\alpha_x})_i (w^{\alpha_x})_{i+1}= \begin{cases}
  w_{i+1} w_i  & \text{ if } i \in T_x \\
  w_i w_{i+1} & \text{ if } i\not\in T_x
\end{cases}. $$
These altogether imply that $w^{\alpha_x}=w+x$.
\end{proof}

\begin{Remark}
Note that if $x=11\ldots11$, then $\alpha_x=\sigma$ in the lemma above.
\end{Remark}

The following result is a kind of generalization of Lemma \ref{dim3} in the presence of an involutory fixed point free 
permutation automorphism. By 
Proposition \ref{dim1}, there is no binary linear code with length $4$ and dimension $3$ 
whose permutation automorphism group is isomorphic to $C_2$. Note also 
that $\PAut(\FF_2^4)=S_4$. As a consequence of these, the result below 
is concerned with codes of length at least $6$.

\begin{Theorem} \label{interval}
Suppose that $C$ is a binary linear code of dimension $k\geq 3$ and length $n=2m$ with $m\geq 3$ 
such that $\PAut(C)=\langle \sigma \rangle$ for $\sigma =(1,2)(3,4) \ldots (n-1, n)$,
then $\dim(F_{\sigma}(C)) \neq k-1$ so that $\lceil k/2 \rceil \leq \dim(F_{\sigma}(C)) \leq k-2$. In particular, 
$k\neq 3$ or equivalently there is no $3$-dimensional linear code $C$ of even length with the 
property that  $\PAut(C)=\langle \sigma \rangle$.  
\end{Theorem}

\begin{proof}
By Lemma \ref{dim3}, it is enough to show that $F_{\sigma}(C)$ can not be $(k-1)$-dimensional. Suppose to the contrary that 
$F_{\sigma}(C)$ is $(k-1)$-dimensional. Let $w\in C\backslash F_{\sigma}(C)$, then $w^{\sigma} \neq w$ and 
$w+w^{\sigma} \in F_{\sigma}(C)$. 

If for some odd $i$, $w_i=w_{i+1}$ then the transposition $(i, i+1)$ fixes $w$ as well as any element of $F_{\sigma}(C)$. 
Moreover, since $F_{\sigma}(C)$ is 
$(k-1)$-dimensional any element of $C$ can be written as $\lambda_1 w+ \lambda_2 c$ where $c\in F_{\sigma}(C)$ 
and $\lambda_1, \lambda_2 \in \FF_2$. 
Hence, $(i, i+1)$ fixes every element of $C$, that is $\PAut(C)\neq \langle \sigma \rangle$, a contradiction. 

Hence, for all odd $i \in \{1, 2, \ldots, n\}$, $w_i\neq w_{i+1}$, and it follows that 
$$(w+w^{\sigma})_i=(w+w^{\sigma})_{i+1}=w_i+w_{i+1}=1,$$ that is $w+w^{\sigma}=y =(1 1 \ldots 11)$. 
Now, since $\dim(F_{\sigma}(C))=k-1 \geq 2$, there is a non-zero element $x \in F_{\sigma}(C)$ 
such that $x \neq y$, so $x\neq (1 1 \ldots 1 1)$. Then $T_x$ is non-empty, since $x$ is a non-zero 
element of $F_{\sigma}(C)$.
Also since $x\neq (1 1 \ldots 1 1)$, we have that $T_x$ is strictly contained $T_y$ and hence the permutation 
$$\alpha_x= \Pi_{i \in T_x} (i, i+1)$$
is an involution in $S_n$ which is different from $\sigma$. Hence, from Lemma \ref{wx}, we 
have that $w^{\alpha_x}=w+x$. Since any element of $C$ 
can be written as $\lambda_1 w+ \lambda_2 c$ where 
$c\in F_{\sigma}(C)$ and $\lambda_1, \lambda_2 \in \FF_2$ and 
that $\alpha_x$ fixes any element of $F_{\sigma}(C)$, we deduce 
that $\alpha_x \in \PAut(C)$, which contradicts with the fact that $\PAut(C)=\langle \sigma \rangle$. 
This finishes the proof.
\end{proof}

\begin{Theorem}\label{dim4}
There is no $4$-dimensional binary linear code $C$ of length $n=2m$ with the property that $\PAut(C)=\langle \sigma \rangle$ where $\sigma =(1,2)(3,4) \ldots (n-1, n)$. 
\end{Theorem}

\begin{proof}
Suppose to the contrary that $C$ is a $4$-dimensional binary linear code of length $n=2m$ with the property that $\PAut(C)=\langle \sigma \rangle$. 
If $n=4$, we have that $C=\FF_2^4$ and $\PAut(C)=S_4$, so we have that $n=2m \geq 6$.
Furthermore, by Theorem \ref{interval}, we have that $\dim(F_{\sigma}(C))=2$. Let $\{x, y\}$ be a basis for $F_{\sigma}(C)$ and extend it to a basis 
$$\CB=\{ x, y, w, u\}$$ 
of $C$. Then $w+w^{\sigma}$ and $u+u^{\sigma}$ lie in $F_{\sigma}(C)$. If $w+w^{\sigma}=u+u^{\sigma}$, then we have that 
$w+u \in F_{\sigma}(C) \cap \langle w, u \rangle=\{0\}$, which implies that $w=u$. So we have that $w+w^{\sigma}\neq u+u^{\sigma}$. 
Hence without loss of generality, assume that $w+w^{\sigma}=x$ and $u+u^{\sigma}=y.$ Then 
$i \in T_x$ if and only if $w_i\neq w_{i+1}$. Similarly, $i \in T_y$ if and only if $u_i\neq u_{i+1}$. 

Note that $T_x \neq T_y$ since $x\neq y$. Set $T:=T_x \cup T_y$. 

{\bf Case 1:} Assume that $T$ is not equal to the set of all odd integers inside $\{1,2, \ldots, n\}$. Set 
$$\alpha=\Pi_{i\in T} (i, i+1).$$
Since $T$ is non-empty and not equal to the set of all odd integers inside $\{1,2, \ldots, n\}$, 
$\alpha$ is an involution different from $\sigma$ and it fixes both $x$ and $y$. Moreover, 
$$(w^{\alpha})_i (w^{\alpha})_{i+1}= \begin{cases}
  w_{i+1} w_i  & \text{ if } i \in T \\
  w_i w_{i+1} & \text{ if } i\not\in T
\end{cases}. $$
Since any $i\in T_x \iff w_i\neq w_{i+1}$, we have that 
$$(w+x)_i (w+x)_{i+1}= \begin{cases}
  w_{i+1} w_i  & \text{ if } i \in T_x \\
  w_{i} w_{i+1}=w_{i+1}w_i  & \text{ if } i \in T\backslash T_x \\
  w_i w_{i+1} & \text{ if } i\not\in T
\end{cases} .$$
Hence, $w^{\alpha}=w+x$. Similarly, it is easy to see that $u^{\alpha}=u+y$. Therefore, $\alpha \in \PAut(C)$ 
which is a contradiction. 

{\bf Case 2:} Assume that $T$ is equal to the set of all odd integers inside $\{1,2, \ldots, n\}$. \\

{\bf Case 2(a):} One of $T_x \backslash T_y$ or $T_y \backslash T_x$ is an empty set (note that if both of these 
sets are empty then $T_x=T_y$ which is not possible). Without loss 
of generality assume that $T_x \backslash T_y =\emptyset$. Then $T_x$ is strictly smaller than $T_y=T$. Consider the involution
$$\alpha_x=\Pi_{i\in T_x} (i, i+1).$$
Then $\alpha_x$ is a non-trivial involution different from $\sigma$.
We have that $\alpha_x$ fixes $x$ and $y$. Since $T_x \subset T_y$, by Lemma \ref{wx}, we have that 
$w^{\alpha_x}=w+x$ and $u^{\alpha_x}=u+x$. Therefore, $\alpha_x \in \PAut(C)$ 
which is a contradiction. \\

{\bf Case 2(b):}  Both $T_x \backslash T_y$ and $T_y \backslash T_x$ has a unique element. Then it follows that $T_x \cap T_y$ 
has $m-2\geq 1$ elements. \\

{\bf Case 2(b) (i):} Assume $m=3$ or equivalently $T_x \cap T_y$ has a unique element. Then the length of the code is $2m=6$ and 
so $\dim(C^{\perp})=6-4=2$. But this contradicts with Theorem \ref{dim12}, since $\PAut(C)=\PAut(C^{\perp})$.\\

 {\bf Case 2(b) (ii):} Assume $T_x \cap T_y=\{i, j\}$ with the property that 
$u_i = u_j$ and $w_i\neq w_j$. Let $T_x \backslash T_y=\{k\}$ and $T_y \backslash T_x=\{l\}$. Set 
 $$\alpha= (k, k+1) (i, j) (i+1, j+1).$$ 
 Then $\alpha$ is a non-trivial involution different from $\sigma$.
Since $x_ix_{i+1}=11=x_j x_{j+1}$ and $y_iy_{i+1}=11=y_j y_{j+1}$, we have that $\alpha$ fixes $x$ and $y$. 
 Moreover, by the assumption on $u$, it follows that $u_{i+1}=u_{j+1}$  and $u_{k}=u_{k+1}$ and so we have that $\alpha$ fixes $u$. 
 Since $k, i , j \in T_x$, we have that $w_k \neq w_{k+1}$ and also $w_i \neq w_j \implies w_j =w_{i+1}$ and $w_{j+1}=w_i$, we have that 
 for any odd $r$
 $$(w^{\alpha})_r (w^{\alpha})_{r+1}= \begin{cases}
  w_{k+1} w_k  & \text{ if } r=k \\
  w_j u_{j+1} & \text{ if }  r=i \\
   w_{i} w_{i+1} & \text{ if } r=j \\
    w_{r} w_{r+1} & \text{ otherwise } 
\end{cases}$$
$$(w+x)_r (w+x)_{r+1}= \begin{cases}
  w_{k+1} w_k  & \text{ if } r=k \\
  w_j u_{j+1} & \text{ if }  r=i \\
   w_{i} w_{i+1} & \text{ if } r=j \\
    w_{r} w_{r+1} & \text{ otherwise } 
\end{cases}. $$
Therefore, we get that 
 $w^{\alpha}=w+x$. So 
 $\alpha\in \PAut(C)$ and we have a contradiction. \\

 {\bf Case 2(b) (iii):} Assume $T_x \cap T_y$ has at least two elements $i$ and $j$ with the property that 
 $u_i=u_j$ and $w_i=w_j$. Set 
 $$\alpha= (i, j)(i+1, j+1).$$ 
 Then since $x_ix_{i+1}=11=x_j x_{j+1}$ and $y_iy_{i+1}=11=y_j y_{j+1}$, we have that $\alpha$ fixes $x$ and $y$. 
 Since $i$ and $j$ satisfy $u_i=u_j$ and $w_i=w_j$, then it follows that $u_{i+1}=u_{j+1}$ and $w_{i+1}=w_{j+1}$ since 
 $u_i\neq u_{i+1}, \ u_j\neq u_{j+1}$. So $\alpha$ fixes both $w$ and $u$. Hence, $\alpha \in \PAut(C)$, a contradiction 
 since $\alpha$ is an involution which is different from $\sigma$.  \\
 
{\bf Case 2(b) (iv):} Assume $T_x \cap T_y$ has at least two elements $i$ and $j$ with the property that 
$u_i\neq u_j$ and $w_i\neq w_j$. Set 
 $$\alpha= (i, j+1)(i+1, j).$$ 
 Then since $x_ix_{i+1}=11=x_j x_{j+1}$ and $y_iy_{i+1}=11=y_j y_{j+1}$, we have that $\alpha$ fixes $x$ and $y$. 
 Since $i$ and $j$ satisfy $u_i\neq u_j$ and $w_i\neq w_j$, then it follows that $u_{i}=u_{j+1}$ and $w_{i}=w_{j+1}$ since 
 $u_i\neq u_{i+1}, \ u_j\neq u_{j+1}$. Similarly, $u_{i+1}=u_{j}$ and $w_{i+1}=w_{j}$. So $\alpha$ fixes both $w$ and $u$. 
 Hence, $\alpha \in \PAut(C)$, a contradiction since $\alpha$ is an involution which is different from $\sigma$.  \\
 
 {\bf Case 2(b) (v):} Assume $T_x \cap T_y$ has at least three elements. Suppose there is no pair of elements which satisfies Case 2(b)(iii) and 
 Case 2(b)(iv). But this is not possible. Indeed, say $i, j, k \in T_x \cap T_y$ such that $u_i =u_j$ and $w_i\neq w_j$. 
 If $u_i=u_k$ and $w_i \neq w_k$, then we have that $u_j=u_k$ and $w_j=w_k$, hence we are in Case 2(d)(ii). If  
  $u_i\neq u_k$ and $w_i = w_k$, then we have that $u_j\neq u_k$ and $w_j\neq w_k$, hence we are in Case 2(d)(iii). 
  So there is a pair of elements which satisfies Case 2(b)(iii) and 
 Case 2(b)(iv) and we are done by using the same methods in these cases. \\

{\bf Case 2(c):}  Assume that one of $T_x \backslash T_y$ or $T_y \backslash T_x$ has at least three elements. Without loss 
of generality assume that $T_x \backslash T_y$ has at least three elements. Among these elements there exist $k, l \in T_x \backslash T_y$ 
such that $u_k u_{k+1}= u_l u_{l+1} $ (since $k, l \not \in T_y$ we have that $u_k u_{k+1}, u_l u_{l+1} \in \{00, 11\}$). Since $k, l \in T_x$ 
we have that $w_k\neq w_{k+1}$ and $w_l\neq w_{l+1}$. Let us choose $k_1 \in \{k, k+1\}$ and $l_1 \in \{l, l+1 \}$ such that 
$w_{k_1}=w_{l_1}$. Then $\alpha=(k_1, l_1)$ fixes $w$ and $u$. It is not difficult to see that $\alpha$ fixes $x$ and $y$ also 
since $k, l\in T_x \backslash T_y$. So $\alpha\in \PAut(C)$ and this is a contradiction since $\alpha$ is an involution 
different from $\sigma$.

Note that the same proof works in the case that $T_x \backslash T_y=\{k, l\}$ has exactly two elements and $u_k u_{k+1}= u_l u_{l+1} $. \\

{\bf Case 2(d):}  Assume that one of $T_x \backslash T_y$ or $T_y \backslash T_x$ has exactly two elements which does not 
satisfy the condition in the previous sentence. Suppose that 
$T_x \backslash T_y=\{k, l\}$ and $u_k u_{k+1}\neq u_l u_{l+1} $. Since $k, l \in T_x$, we have that $x_k x_{k+1}=11=x_l x_{l+1}$ and $w_k \neq w_{k+1}$ and $w_l\neq w_{l+1}$. 
Also, since $k, l \not \in T_y$, we have that $y_k y_{k+1}=00=y_l y_{l+1}$ and $u_k  u_{k+1}, u_l u_{l+1} \in \{00, 11\}$. 
Let 
$k_1 \in \{k, k+1\}$ and $l_1 \in \{l, l+1\}$ such that $w_{k_1}=w_{l_1}$ so that $w_{\sigma(k_1)}=w_{\sigma(l_1)}$. Consider the involution
$$\alpha= \big(\prod_{i \in T_y \backslash T_x}(i, i+1)\big) (k_1, l_1)(\sigma(k_1), \sigma(l_1)).$$
It is easy to see that $\alpha$ is a non-trivial involution which is different from $\sigma$.
Moreover, $\alpha$ fixes $x$ and $y$.
Since $w_i=w_{i+1}$ for all $i\in T_y \backslash T_x$, it is easy to see that $\alpha$ fixes $w$. 
For any odd $r$, we have that 
$$(u^{\alpha})_r (u^{\alpha})_{r+1}= \begin{cases}
  u_{r+1} u_r  & \text{ if } r \in T_y \backslash T_x \\
  u_r u_{r+1} & \text{ if }  r \in T_x\cap T_y \\
   u_{l} u_{l+1}=u_{l+1}u_l  & \text{ if } r=k \\
   u_{k} u_{k+1}=u_{k+1}u_k  & \text{ if } r=l 
\end{cases}. $$
Note that for odd $r$, we have that $(x+y)_r(x+y)_{r+1}=00$ if and only if $r \in T_x \cap T_y$. Also
for $r \in T_y \backslash T_x$, we have that $u_r \neq u_{r+1}$ and $u_k u_{k+1}+11=u_l u_{l+1}$ 
since $u_k u_{k+1}\neq u_l u_{l+1} $. As a result, we have that
$$(u+x+y)_r (u+x+y)_{r+1}= \begin{cases}
  u_{r+1} u_{r}  & \text{ if } r \in T_y \backslash T_x  \\
    u_{r} u_{r+1} & \text{ if } r \in T_x\cap T_y \\
   u_{l} u_{l+1} & \text{ if } r=k \\
   u_{k} u_{k+1}  & \text{ if } r=l
\end{cases} .$$
Hence, we deduce that $u^{\alpha}=u+x+y$ and it follows that $\alpha \in \PAut(C)$. However, 
$\alpha$ is an involution different from $\sigma$ so we have a contradiction. \\

Since we get a contradiction in all possible cases, the proof is finished.
\end{proof}
 
 \begin{Corollary}\label{sigma}
 There is no binary linear code $C$ of dimension $k$ and length $n=2m$ with $m\geq 2$ 
 such that $\PAut(C)=\langle \sigma \rangle$ where 
 $k=i$ or $k=n-i$ for $i \in \{1, 2, 3, 4 \}$ and $\sigma =(1,2)(3,4) \ldots (n-1, n)$.
 \end{Corollary}
 
 \begin{proof}
 Follows from Proposition \ref{dim1}, Theorem \ref{dim2}, Theorem \ref{interval} and 
 Theorem \ref{dim4} since $\PAut(C)=\PAut(C^{\perp})$ and $\dim(C)=n-\dim(C^{\perp})$.
 \end{proof}
 
  \begin{Corollary}\label{quasi}
 There is no binary quasi group code $C$ of dimension $k$ and length $n=2m$ with $m\geq 2$ 
 such that $\PAut(C) \cong C_2$ where 
 $k=i$ or $k=n-i$ for $i \in \{1, 2, 3, 4 \}$. 
  \end{Corollary}
 
 \begin{proof}
 Suppose that $C$ is a binary quasi group code such that $\PAut(C) \cong C_2$.
 Then $\PAut(C)=\langle \gamma\rangle $ where $\gamma$ is a fixed point free involution in $S_n$. 
 Then $\gamma$ and $\sigma$ are conjugate elements of $S_n$, that is $\sigma= \beta^{-1} \gamma \beta$ for some $\beta \in S_n$. 
 Let $C_1=C^{\beta}$, then $\PAut(C_1)=\langle \sigma \rangle$. But then Corollary \ref{sigma} gives us a contradiction.  
 \end{proof}

\begin{Corollary}
 There is no binary linear code $C$ of length $2<n=2m\leq 8$ with $\PAut(C)=\langle \sigma \rangle$ where $\sigma =(1,2)(3,4) \ldots (n-1, n)$.
 \end{Corollary}
 
 \begin{proof}
 Follows from Corollary \ref{sigma}.
 \end{proof}
 
  \begin{Corollary}\label{length8}
 There is no binary quasi group code $C$ of length $2< n=2m\leq 8$ with $\PAut(C)\cong C_2$. 
  \end{Corollary}
 
 \begin{proof}
 Follows from Corollary \ref{quasi}.
 \end{proof}
 
 Our computer calculations give us the impression that the generalization of Corollary \ref{length8} may be possible. 
 But we can not give a proof of it, as can be seen from the proof of Theorem \ref{dim4} there will be many cases to consider. 
 We are claiming the following is true.
 
 \begin{Conjecture}
 There is no binary quasi group code $C$ of length $n=2m >2$ such that $\PAut(C)\cong C_2$. 
 \end{Conjecture}

\section{A generalization and an application to some putative extremal self-dual binary codes}

In this section, we present a generalization of the ideas that are used in the proof of Theorem \ref{dim4}. By applying 
this generalization to the putative extremal self-dual binary $[72, 36, 16]$ code and putative extremal self-dual binary 
$[96, 48, 20]$ code, we obtain results concerning the structure of these codes.   

Let $C$ be a binary linear code of length $n=2m$ and $\sigma= (1,2)(3,4) \ldots (n-1, n)\in \PAut(C)$. We 
have that the fixed subcode $F_{\sigma}(C)$ is a subcode of $C$. So there 
is a subcode $N_{\sigma}(C)$ such that 
$$C= F_{\sigma}(C) \oplus N_{\sigma}(C).$$ 
Let $\{ w_{(1)}, \ldots, w_{(k-f)}\}$ be a basis of $N_{\sigma}(C)$ where $k=\dim(C)$ and $f=\dim(F_{\sigma}(C))$. 
Then by Lemma \ref{halfdim}, we have that $f\geq (k-f)$. On the other hand, by the proof of Lemma \ref{halfdim}, 
there is a linearly independent subset $\{ x_{(1)}, \ldots, x_{(k-f)}\}$ of $F_{\sigma}(C)$ where $x_{(j)}=w_{(j)}+w_{(j)}^{\sigma}$ 
for $j=1, \ldots, k-f$. Now extend $\{ x_{(1)}, \ldots, x_{(k-f)}\}$ to a basis $\{ x_{(1)}, \ldots, x_{(f)}\}$ of $F_{\sigma}(C)$. 
Set $$T(\sigma):=T_{x_{(1)}} \cup T_{x_{(2)}} \cup \ldots \cup T_{x_{(k-f)}}$$
where $T_{x_{(j)}}$ is defined as in Section $4$. 
Then note that if $i\not\in T(\sigma)$ then $(x_{(j)})_{i}(x_{(j)})_{i+1}=00$ for any $j=1, \ldots, k-f$. So if $i\not\in T(\sigma)$
the equation  $x_{(j)}=w_{(j)}+w_{(j)}^{\sigma}$ implies that $(w_{(j)})_{i}=(w_{(j)})_{i+1}$ for any $j=1, \ldots, k-f$. 

\begin{Theorem}\label{general}
Let $C$ be a binary linear code of length $n=2m$ where $m\geq 2$ 
such that $\sigma=(1, 2)(3, 4)\ldots (n-1, n) \in \PAut(C)$. Suppose that $T(\sigma) \neq \{1, 3, \ldots, n-1\}$.
Then we have that $\PAut(C)\neq \langle \sigma \rangle$. In particular, $\PAut(C)$ has automorphisms which have some fixed points and 
that $C_2 \times C_2 \leq \PAut(C)$.
\end{Theorem}

\begin{proof}
Let $C$ be a binary linear code of dimension $k$. Let $\beta =\Pi_{i\not\in T_{\sigma}} (i, i+1)$.
Since $T(\sigma)\neq \{1, 3, \ldots, n-1\}$, the permutation $\beta$ is a non-trivial involution. 
Also if $i\not\in T_{\beta}$ then $(x_{(j)})_{i}(x_{(j)})_{i+1}=00$  and 
$(w_{(j)})_{i}=(w_{(j)})_{i+1}$for any $j=1, \ldots, k-f$ . Hence $\beta$ 
fixes all elements in $F_{\sigma}(C)$ and $N_{\sigma}(C)$ so it lies inside 
$\PAut(C)$. It is clear that $\beta$ has some fixed points and $\langle \sigma, \beta \rangle \cong C_2 \times C_2$
\end{proof}

\begin{Corollary} 
Let $C$ be a binary self-dual $[72, 36, 16]$ code. If $\sigma=(1, 2)(3, 4)\ldots (71, 72) \in \PAut(C)$, then $T(\sigma)$ is 
equal to the set of all odd integers between $1$ and $71$. 
\end{Corollary}

\begin{proof}
Suppose to the contrary that $T(\sigma)\neq \{1, 3, \ldots, 71\}$. Then by Theorem \ref{general}, there is a an involution 
in $\PAut(C)$ which is not fixed point free. But this contradicts with \cite[Theorem 5.3]{Bou2002}. So $T_{\sigma} = \{1, 3, \ldots, 71\}$. 

\end{proof}

\begin{Corollary} 
Let $C$ be a binary self-dual $[96, 48, 20]$ code. If $\sigma=(1, 2)(3, 4)\ldots (95, 96) \in \PAut(C)$, then $T(\sigma)$ is 
equal to the set of all odd integers between $1$ and $95$. 
\end{Corollary}

\begin{proof}
Suppose to the contrary that $T(\sigma) \neq \{1, 3, \ldots, 95\}$. Then by Theorem \ref{general}, there is a an involution 
in $\PAut(C)$ which is not fixed point free. But this contradicts with \cite[Theorem 8]{Bou2000}. So $T(\sigma) = \{1, 3, \ldots, 95\}$. 

\end{proof}

\end{document}